\documentclass{amsart}
\usepackage{amsmath,amsfonts,amssymb}
\usepackage{graphicx,amsmath, amscd}
\usepackage{color,comment}
\usepackage[T1]{fontenc}
\usepackage[latin1]{inputenc}

\usepackage[lmargin=2cm,rmargin=2cm,tmargin=2cm,bmargin=2cm]{geometry}
\usepackage[toc,page]{appendix}

\newtheorem{theorem}{Theorem}[section]
\newtheorem{lemma}[theorem]{Lemma}
\newtheorem{corollary}[theorem]{Corollary}
\newtheorem{proposition}[theorem]{Proposition}
\newtheorem{definition}[theorem]{Definition}
\newtheorem{remark}[theorem]{Remark}
\newtheorem{defi/prop}[theorem]{Definition/Proposition}

\newtheorem*{theorem*}{Theorem}

\newcommand{\N}{\mathbf{N}}

\newcommand{\R}{\mathbf{R}}
\newcommand{\C}{\mathbf{C}}
\newcommand{\cF}{\mathcal{F}}
\renewcommand{\P}{\mathbf{P}}

\newcommand{\e}{\varepsilon}
\renewcommand{\leq}{\leqslant}
\renewcommand{\geq}{\geqslant}

\newcommand{\st}{\  : \ }

\newcommand{\Id}{\mathrm{Id}}
\newcommand{\cH}{\mathcal{H}}

\DeclareMathOperator{\conv}{conv}

\DeclareMathOperator{\card}{card}

\DeclareMathOperator{\tr}{Tr}
\DeclareMathOperator{\Sym}{Sym}
\DeclareMathOperator{\E}{\mathbf{E}}

\newcommand{\scalar}[2]{\langle #1 , #2\rangle}

\newcommand{\ketbra}[2]{| #1 \rangle \langle #2 |}
\newcommand{\bra}[1]{\langle #1 |}
\newcommand{\ket}[1]{| #1 \rangle}

\title{Zonoids and Sparsification of Quantum Measurements}

\author{Guillaume Aubrun}
\address{Institut Camille Jordan, Universit\'e Claude Bernard Lyon 1, 43 boulevard du 11 novembre 1918,
69622 Villeurbanne cedex, France}
\email{aubrun@math.univ-lyon1.fr}
\author{Cécilia Lancien}
\address{Institut Camille Jordan, Universit\'{e} Claude Bernard Lyon 1, 43 boulevard du 11 novembre 1918, 69622
Villeurbanne Cedex, France and
F\'{\i}sica Te\`{o}rica: Informaci\'{o} i Fenomens Qu\`{a}ntics, Universitat Aut\`{o}noma de Barcelona,
ES-08193 Bellaterra (Barcelona), Spain}
\email{lancien@math.univ-lyon1.fr}


\begin{document}

\begin{abstract}
In this paper, we establish a connection between zonoids (a concept from classical convex geometry) and the distinguishability norms associated
to quantum measurements or POVMs (Positive Operator-Valued Measures), recently introduced in quantum information theory.

This correspondence allows us to state and prove the POVM version of classical results from the local theory of Banach spaces about the approximation
of zonoids by zonotopes. We show that on $\C^d$, the uniform POVM (the most symmetric POVM) can be sparsified, i.e. approximated by a discrete POVM
having only $O(d^2)$ outcomes. We also show that similar (but weaker) approximation results actually hold for any POVM on $\C^d$.

By considering an appropriate notion of tensor product for zonoids, we extend our results to the multipartite setting: we show,
roughly speaking, that local POVMs may be sparsified locally. In particular, the local uniform POVM on $\C^{d_1}\otimes\cdots\otimes\C^{d_k}$
can be approximated by a discrete POVM which is local and has $O(d_1^2 \times \cdots \times d_k^2)$ outcomes.
\end{abstract}

\thanks{This research was supported by the ANR project OSQPI  ANR-11-BS01-0008}

\maketitle

\section*{Introduction}

A classical result by Lyapounov (\cite{RudinFA}, Theorem 5.5) asserts that the range of a non-atomic $\R^n$-valued vector measure is closed and convex.
Convex sets in $\R^n$ obtained in this way are called zonoids. Zonoids are equivalently characterized as
convex sets which can be approximated by finite sums of segments.

In this paper we consider a special class of vector measures: Positive Operator-Valued Measures (POVMs). In the formalism of quantum mechanics, POVMs
represent the most general form of a quantum measurement. Recently, Matthews, Wehner and Winter \cite{MWW} introduced
the distinguishability norm associated to a POVM. This norm has an operational interpretation as the bias of the POVM for the state
discrimination problem (a basic task in quantum information theory) and is closely related to the zonoid arising from Lyapounov's theorem.

A well-studied question in high-dimensional convexity is the approximation of zonoids by zonotopes. The series of papers
\cite{FLM,Schechtman,BLM,Talagrand}
culminates in the following result: any zonoid in $\R^n$ can be approximated by the sum of $O(n \log n)$ segments. The aforementioned connection
between POVMs and zonoids allows us to state and prove approximation results for POVMs, which improve on previously known bounds. Precise statements
appear as Theorem \ref{theorem:approximation-of-U} and \ref{theorem:approximation-any}.

Our article is organized as follows. Section \ref{section:POVM} introduces POVMs and their associated distinguishability norms.
Section \ref{section:POVMs-zonoids} connects POVMs with zonoids. Section \ref{section:tensorizing} introduces a notion of tensor product
for POVMs, and the corresponding notion for zonoids.
Section \ref{section:sparsification} pushes forward this connection to state the POVM
version of approximation results for zonoids, which are proved in Sections \ref{sec:uniform-POVM}, \ref{section:bernstein} and \ref{sec:approximation-any}. Section
\ref{section:sparsification-multipartite} provides sparsification results for local POVMs on multipartite systems.

The reader may have a look at Table \ref{table:analogies}, which summarizes analogies between zonoids and POVMs.

\subsection*{Notation}

We denote by $\cH(\C^d)$ the space of Hermitian operators on $\C^d$, and by $\cH_+(\C^d)$ the subset of positive operators. We denote by $\|\cdot\|_1$
the trace class norm,
by $\|\cdot\|_{\infty}$ the operator norm
and by $\|\cdot\|_{2}$ the Hilbert--Schmidt norm. Notation $[-\Id,\Id]$ stands for the set of self-adjoint operators $A$ such that $-\Id \leq A \leq \Id$. In
other words $[-\Id,\Id]$ is the self-adjoint part of the  unit ball for
$\|\cdot\|_{\infty}$. We denote by $S(\C^d)$ the set of states on $\C^d$ (a state is a positive operator with trace $1$).

Let us recall a few standard concepts from classical convex geometry that we will need throughout our proofs.
The support function $h_K$ of a convex compact set $K \subset \R^n$
is the function defined for $x \in \R^n$ by $ h_K(x) = \sup \{ \langle x,y \rangle \st y \in K \}$.
Moreover, for a pair $K,L$ of convex compact sets, the inclusion $K \subset L$ is equivalent to the inequality $h_K \leq h_L$.
The polar of a convex set $K \subset \R^n$ is $K^\circ = \{ x \in \R^n \st \langle x,y \rangle \leq 1 \textnormal{ whenever } y \in K \}$.
The bipolar theorem (see e.g. \cite{Barvinok}) states that $(K^\circ)^\circ$ is the closed convex hull of $K$ and $\{0\}$.
A convex body is a convex compact set with non-empty interior.
Whenever we apply tools from convex geometry in the (real) space $\cH(\C^d)$ (e.g. polar or support function),
we use the Hilbert--Schmidt inner product $(A,B) \mapsto \tr AB$ to define the Euclidean structure.

The letters $C,c,c_0,\dots$ denote numerical constants, independent from any other parameters such as the dimension. The value of these constants may
change from occurrence to occurrence. Similarly $c(\e)$ denotes a constant depending only on the parameter $\e$. We also use the following
convention: whenever a formula is given for the dimension of a (sub)space, it is tacitly understood that one should take the integer part.

\section{POVMs and distinguishability norms} \label{section:POVM}

In quantum mechanics, the state of a $d$-dimensional system is described by a positive operator on $\C^d$ with trace $1$.
The most general form of a
measurement that may be performed on such a quantum system is encompassed by the formalism of Positive Operator-Valued Measures (POVMs).
Given a set $\Omega$ equipped with a $\sigma$-algebra $\mathcal{F}$, a POVM on $\C^d$
is a map $\mathrm{M} :\cF \to \cH_+(\C^d)$  which is $\sigma$-additive and such that $\mathrm{M}(\Omega) = \Id$.
In this definition the space $(\Omega,\cF)$ could potentially be infinite, so that the POVMs defined on it would be continuous. However, we often restrict ourselves to the subclass
of discrete POVMs, and a main point of this article is to substantiate this ``continuous to discrete'' transition.

A discrete POVM is a POVM in which the underlying $\sigma$-algebra $\cF$ is required to be finite. In that case there is a finite
partition $\Omega = A_1 \cup \cdots \cup A_n$ generating $\cF$. The positive operators $M_i = \mathrm{M}(A_i)$ are often referred to as the elements of
the POVM, and they satisfy the condition $M_1 + \dots + M_n = \Id$. We usually identify a discrete POVM with the set of its elements by writing
$\mathrm{M} = (M_i)_{1 \leq i \leq n}$. The index set $\{1,\dots,n\}$ labels the outcomes of the measurement.
The integer $n$ is thus the number of outcomes of $\mathrm{M}$ and can be seen as a crude way to measure the
complexity of $\mathrm{M}$.

What happens when measuring with a POVM $\mathrm{M}$ a quantum system in a state $\rho$ ?
In the case of a discrete POVM $\mathrm{M}=(M_i)_{1 \leq i \leq n}$, we know from Born's rule that the outcome $i$ is output with probability
$\tr (\rho M_i)$.
This simple formula can be used to quantify the efficiency of a POVM to perform the task of state discrimination. State discrimination can be
described as follows: a quantum system is prepared in an unknown state which is either $\rho$ or $\sigma$ (both hypotheses being a priori
equally likely), and we have to guess the unknown state. After measuring it with the discrete POVM $\mathrm{M} =
(M_i)_{1 \leq i \leq n}$, the optimal strategy,
based on the maximum likelihood probability, leads to a probability of wrong guess equal to \cite{Holevo,Helstrom}
\[ \P_{error} = \frac{1}{2}\left( 1 - \frac{1}{2} \sum_{i=1}^n \left| \tr (\rho M_i) - \tr(\sigma M_i) \right| \right) . \]
In this context, the quantity $\frac{1}{2} \sum_{i=1}^n \left| \tr (\rho M_i) - \tr(\sigma M_i) \right|$
is therefore called the bias of the POVM $\mathrm{M}$ on the state pair $(\rho,\sigma)$.

Following \cite{MWW}, we introduce a norm on $\cH(\C^d)$, called the distinguishability norm associated to $\mathrm{M}$,
and defined for $\Delta \in \cH(\C^d)$ by
\begin{equation} \label{eq:definition-norm-discrete} \|\Delta\|_{\mathrm{M}} = \sum_{i=1}^n \left| \tr (\Delta M_i) \right|. \end{equation}
It is such that $\P_{error} = \frac{1}{2}\left(1 - \frac{1}{2} \| \rho - \sigma \|_{\mathrm{M}}\right)$, and thus quantifies how powerful the POVM
$\mathrm{M}$ is in discriminating one state from another with the smallest probability of error.

The terminology ``norm'' is slightly abusive since one may have $\|\Delta\|_{\mathrm{M}}=0$
for a nonzero $\Delta \in \cH(\C^d)$. The functional $\|\cdot\|_{\mathrm{M}}$ is however always a semi-norm, and
it is easy to check that $\|\cdot\|_{\mathrm{M}}$ is a norm if and only if the POVM elements $(M_i)_{1 \leq i \leq n}$ span $\cH(\C^d)$ as a vector
space. Such POVMs
are called informationally complete in the quantum information literature.

Similarly, the distinguishability norm associated to a general POVM $\mathrm{M}$, defined on a set $\Omega$ equipped with a $\sigma$-algebra
$\mathcal{F}$, is described for $\Delta \in \cH(\C^d)$
by
\begin{equation}
\label{eq:definition-norm-continuous}
\| \Delta \|_{\mathrm{M}} = \| \tr(\Delta \mathrm{M}(\cdot)) \|_{\mathrm{TV}} = \sup_{A \in \cF} \big[ \tr(\Delta \mathrm{M}(A)) -
\tr(\Delta \mathrm{M}(\Omega \setminus A)) \big]  = \sup_{M \in \mathrm{M}(\mathcal{F})} \tr (\Delta(2M-\Id)).\end{equation}
Here $\|\mu\|_{\mathrm{TV}}$ denotes the total variation of a measure $\mu$. When $\mathrm{M}$ is discrete, formulae
\eqref{eq:definition-norm-discrete} and
\eqref{eq:definition-norm-continuous} coincide.
Note also that the inequality $\|\cdot\|_{\mathrm{M}} \leq \|\cdot\|_1$ holds for any POVM $\mathrm{M}$, with equality on $\cH_+(\C^d)$.

Given a POVM $\mathrm{M}$, we denote by $B_{\mathrm{M}} = \{ \|\cdot \|_{\mathrm{M}} \leq 1 \}$ the unit ball for the distinguishability norm,
and $K_{\mathrm{M}}= (B_{\mathrm{M}})^\circ$ its polar, i.e.
\[ K_{\mathrm{M}} = \{ A \in \cH(\C^d) \st \tr (AB) \leq 1 \textnormal{ whenever } \|B\|_{\mathrm{M}} \leq 1 \}. \]
The set $K_{\mathrm{M}}$ is a compact convex set. Moreover $K_{\mathrm{M}}$ has nonempty interior if and only if the POVM $\mathrm{M}$
is informationally complete. It follows from the inequality $\|\cdot\|_{\mathrm{M}} \leq \|\cdot\|_1$ that $K_{\mathrm{M}}$ is always included in the
operator interval $[-\Id,\Id]$.

On the other hand, it follows from \eqref{eq:definition-norm-continuous} that $B_{\mathrm{M}} = ( 2 \mathrm{M}(\mathcal{F}) - \Id)^{\circ}$,
and the bipolar theorem implies that
\begin{equation} \label{eq:K_M-polar} K_{\mathrm{M}} = 2 \conv ( \mathrm{M}(\mathcal{F})) - \Id . \end{equation}
By Lyapounov's theorem,  the convex hull operation is not needed when $\mathrm{M}$ is non-atomic.
For a discrete POVM $\mathrm{M} = ( M_i )_{1 \leq i \leq n}$, equation \eqref{eq:K_M-polar} may be rewritten in the form
\begin{equation} \label{eq:K_M-zonotope} K_{\mathrm{M}} = \conv \{\pm M_1 \} + \cdots + \conv \{ \pm M_n \}, \end{equation}
where the addition of convex sets should be understood as the Minkowski sum: $A+B = \{a+b \st a \in A,\ b \in B \}$.

We are going to show that POVMs can be sparsified, i.e approximated by discrete POVMs with few outcomes.
The terminology ``approximation'' here refers to the associated distinguishability norms: a POVM $\mathrm{M}$ is considered to be ``close'' to a
POVM $\mathrm{M}'$ when
their distinguishability norms satisfy inequalities of the form
\[ (1-\e)\|\cdot\|_{\mathrm{M}'} \leq \|\cdot\|_{\mathrm{M}} \leq (1+\e)\|\cdot\|_{\mathrm{M}'}. \]
This notion of approximation has an operational significance: two POVMs are comparable when both lead to comparable biases when used for any state
discrimination task. Let us perhaps stress that point: if one has additional information on the states to be discriminated,
it may of course be used to design a POVM specifically efficient for those (one could for instance be interested in the problem
of distinguishing pairs of low-rank states \cite{Sen,AE}).

In this paper, we study the distinguishability norms from a functional-analytic point of view. We are mostly
interested in the asymptotic regime,
when the dimension $d$ of the underlying Hilbert space is large.

\begin{table}
\label{table:analogies}
\renewcommand{\arraystretch}{1.5}
\small{
\begin{tabular}{|c|c|}
\hline
Zonotope which is the Minkowski sum of $N$ segments & Discrete POVM with $N$ outcomes \\
\hline
Zonoid = limit of zonotopes  & General POVM = limit of discrete POVMs \\
\hline
Tensor product of zonoids  & Local POVM on a multipartite system \\
\hline
Euclidean ball $B_2^n$ & Uniform POVM $\mathrm{U}_d$ \\
= most symmetric zonoid in $\R^n$ & = most symmetric POVM on $\C^d$ \\
\hline
``4th moment method'' (\cite{Rudin60}, explicit): $c B_2^n \subset Z \subset C B_2^n$,  & ``Approximate $4$-design POVM'' \cite{AE}: \\
with $Z$ a zonotope which is the sum of $O(n^2)$ segments. & explicit sparsification of $\mathrm{U}_d$ with $O(d^4)$ outcomes. \\
\hline
Measure concentration (\cite{FLM}, non-explicit): $(1-\e) B_2^n \subset Z \subset (1+\e) B_2^n$, & Theorem \ref{theorem:approximation-of-U}:
a randomly chosen POVM  \\
with $Z$ a zonotope which is the sum of $O_\e(n)$ segments. & with $O(d^2)$ outcomes is a sparsification of $\mathrm{U}_d$. \\
\hline
Derandomization \cite{GW,LS,IS} & ? \\
\hline
Any zonoid in $\R^n$ can be approximated by a zonotope & Theorem \ref{theorem:approximation-any}: any POVM on $\C^d$ can be sparsified \\
which is the sum of $O(n \log n)$ segments \cite{Talagrand}. & into a sub-POVM with $O(d^2 \log d)$ outcomes. \\
\hline
\end{tabular}
}
\caption{A ``dictionary'' between zonoids and POVMs}
\end{table}

\section{POVMs and zonoids}

\label{section:POVMs-zonoids}

\subsection{POVMs as probability measures on states}

The original definition of a POVM involves an abstract measure space, and the specification of this measure space is irrelevant when considering the
distinguishability norms. The following proposition, which is probably well-known, gives a more concrete look at POVMs as probability measures
on the set $S(\C^d)$ of states on $\C^d$.

\begin{proposition} \label{proposition:POVM-states}
Let $\mathrm{M}$ be a POVM on $\C^d$. There is a unique Borel probability measure $\mu$ on $S(\C^d)$ with barycenter equal to $\Id/d$ and such
that, for any $\Delta \in \cH(\C^d)$,
\begin{equation} \label{eq:support-function-POVM} \|\Delta\|_{\mathrm{M}} = d \int_{S(\C^d)} \left| \tr(\Delta \rho)\right| \,
\mathrm{d} \mu(\rho). \end{equation}
Conversely, given a Borel probability measure $\mu$ with barycenter equal to $\Id/d$, there is a POVM $\mathrm{M}$ such that
\eqref{eq:support-function-POVM} is satisfied.
\end{proposition}

\begin{proof}
We use the polar decomposition for vector measures, which follows from applying the Radon--Nikodym theorem to vector measures (see \cite{RudinRCA}, Theorem 6.12): a vector measure $\mu$ defined on a $\sigma$-algebra
$\mathcal{F}$ on $\Omega$ and
taking values in a normed space $(\R^n,\|\cdot\|)$ satisfies $d\mu = h d|\mu|$ for some measurable function $h : \Omega\to \R^n$. Moreover, one
has $\|h\| = 1$ $|\mu|$-a.e. Here $|\mu|$ denotes the total variation measure of $\mu$.

Let $\mathrm{M}$ be a POVM on $\C^d$, defined on a $\sigma$-algebra $\mathcal{F}$ on $\Omega$. We equip $\cH(\C^d)$ with the trace norm, so that we simply
have $|\mathrm{M}| = \tr \mathrm{M}$ and $|\mathrm{M}|(\Omega)=d$. The polar decomposition yields a measurable function $h : \Omega \to \cH(\C^d)$
such that $\|h\|_1 = 1$ $|\mathrm{M}|$-a.e. Moreover, the fact that $\mathrm{M}(\mathcal{F})\subset\cH_+(\C^d)$ implies that $h\in\cH_+(\C^d)$ $|\mathrm{M}|$-a.e. Let $\mu$ be the push forward of $\frac{1}{d} |\mathrm{M}|$
under the map $h$. We have
\[ \Id = \mathrm{M}(\Omega) = \int_{\Omega} h \, \mathrm{d}|\mathrm{M}| = d \int_{\cH(\C^d)} \rho \, \mathrm{d}\mu(\rho).\]
And since $h\in S(\C^d)$ a.e., $\mu$ is indeed
a Borel probability measure on $S(\C^d)$, with barycenter equal to $\Id/d$. Finally, for any $\Delta \in \mathcal{H}(\C^d)$,
\[ \|\Delta\|_{\mathrm{M}} = \int_\Omega |\tr (\Delta h)| \, \mathrm{d} |\mathrm{M}| = d \int_{S(\C^d)} |\tr (\Delta \rho) |  \, \mathrm{d} \mu(\rho) .\]
We postpone the proof of uniqueness to the next subsection (see after Proposition \ref{proposition:POVM}).

Conversely, given a Borel probability measure $\mu$ on $S(\C^d)$ with barycenter at $\Id/d$, consider the vector measure
$\mathrm{M}:\mathcal{B} \to \cH(\C^d)$, where $\mathcal{B}$ is the Borel $\sigma$-algebra on $S(\C^d)$, defined by
\[ \mathrm{M}(A) = d \int_A \rho \, \mathrm{d} \mu (\rho) .\]
It is easily checked that $\mathrm{M}$ is a POVM and that formula \eqref{eq:support-function-POVM} is satisfied.
\end{proof}

Note that in the case of a discrete POVM $\mathrm{M} = (M_i)_{1 \leq i \leq n}$, the corresponding probability measure is
\[ \mu = \frac{1}{d} \sum_{i=1}^n \left( \tr M_i \right) \, \delta_{\frac{M_i}{\tr M_i}} .\]

\begin{corollary} \label{corollary:approximation}
Given a POVM $\mathrm{M}$ on $\C^d$, there is a sequence $(\mathrm{M}_n)$ of discrete POVMs such that $K_\mathrm{M_n}$ converges to $K_\mathrm{M}$
in Hausdorff distance. Moreover, if $\mu$ (resp. $\mu_n$) denotes the probability measure on $S(\C^d)$ associated to $\mathrm{M}$
(resp. to $\mathrm{M}_n$)
as in \eqref{eq:support-function-POVM}, we can guarantee that the support of $\mu_n$ is contained into the support of $\mu$.
\end{corollary}

\begin{proof}
Let $\mu$ be the probability measure associated to $\mathrm{M}$.
Given $n$, let $(Q_k)$ be a finite partition of $S(\C^d)$ into sets of diameter at most $1/n$ with respect to the trace norm.
Let $\rho_k \in S(\C^d)$ be the barycenter of the restriction of $\mu$ to $Q_k$
(only defined when $\mu(Q_k)>0$). The probability measure
\[ \mu_n = \sum_{k} \mu(Q_k) \delta_{\rho_k} \]
has the same barycenter as $\mu$, and the associated POVM $\mathrm{M}_n$ satisfies
\[ \left| h_{K_\mathrm{M}} (\Delta) - h_{K_{\mathrm{M}_n}} (\Delta) \right| \leq d \frac{\|\Delta\|_{\infty}}{n} ,\]
and therefore $K_\mathrm{M_n}$ converges to $K_\mathrm{M}$.

The condition on the supports can be enforced by changing slightly the definition
of $\mu_n$. For each $k$ we can write $\rho_k = \sum \lambda_{k,j} \rho_{k,j}$, where $(\lambda_{k,j})$ is a convex
combination and $(\rho_{k,j})$ belong to the support of $\mu$ restricted to $Q_k$. The measure
\[ \mu'_n = \sum_{k} \mu(Q_k) \sum_j \lambda_{k,j} \delta_{\rho_{k,j}} \]
satisfies the same properties as $\mu_n$, and its support is contained into the support of $\mu$.
\end{proof}

\subsection{POVMs and zonoids}

We connect here POVMs with zonoids, which form an important family of convex bodies (see \cite{Bolker,SW,GW} for surveys on zonoids
to which we refer for all the material presented here).
A zonotope $Z \subset \R^n$ is a closed convex set which can be written as the Minkowski sum of finitely many segments,
i.e. such that there exist finite sets of vectors
$(u_i)_{1 \leq i \leq N}$ and $(v_i)_{1 \leq i \leq N}$ in $\R^n$ such that
\begin{equation} \label{eq:def-zonotope} Z = \conv \{ u_1,v_1 \} + \cdots + \conv \{ u_N,v_N \} .\end{equation}
A zonoid is a closed convex set which can be approximated by zonotopes (with respect to the Hausdorff distance).
Every zonoid has a center of symmetry, and therefore can be translated into a (centrally) symmetric zonoid.
Note that for a centrally symmetric zonotope, we can choose $v_i=-u_i$ in \eqref{eq:def-zonotope}.

Here are equivalent characterizations of zonoids.

\begin{proposition} \label{proposition:zonoids}
Let $K \subset \R^n$ be a symmetric closed convex set. The following are equivalent.
\begin{enumerate}
\item[(i)] $K$ is a zonoid.
\item[(ii)] There is a Borel positive measure $\nu$ on the sphere $S^{n-1}$ which is even (i.e. such that $\nu(A)=\nu(-A)$ for any Borel set
$A \subset S^{n-1}$) and such that, for every $x\in\R^n$,
\begin{equation} \label{eq:support-function-zonoid} h_K(x) = \int_{S^{n-1}} | \langle x,\theta \rangle | \, \mathrm{d}\nu(\theta) . \end{equation}
\item[(iii)] There is a vector measure $\mu : (\Omega,\cF) \to \R^n$ such that $K = \mu(\cF)$.
\end{enumerate}
Moreover, when these conditions are satisfied, the measure $\nu$ is unique.
\end{proposition}

\begin{remark}
Having the measure $\nu$ supported on the sphere and be even is only a matter of
normalization and a way to enforce uniqueness: if $\nu$ is a Borel measure on $\R^n$ for which linear forms are integrable,
there is a symmetric zonoid $K \subset \R^n$ such that
 \[ h_K(x) = \int_{\R^n} | \langle x, y \rangle | \, \mathrm{d}\nu(y) . \]
\end{remark}

As an immediate consequence, we characterize which subsets of $[-\Id,\Id]$ arise as $K_\mathrm{M}$ for some POVM $\mathrm{M}$.

\begin{proposition} \label{proposition:POVM}
Let $K \subset \cH(\C^d)$ be a symmetric closed convex set. Then the following are equivalent.
\begin{enumerate}
\item[(i)] $K$ is a zonoid such that $K \subset [-\Id,\Id]$ and $\pm \Id \in K$.
\item[(ii)] There exists a POVM $\mathrm{M}$ on $\C^d$ such that $K=K_{\mathrm{M}}$.
\end{enumerate}
Moreover, $K$ is a zonotope only if the POVM $\mathrm{M}$  can be chosen to be discrete.
\end{proposition}

\begin{proof}
Let $K$ be a zonoid such that $\pm \Id \in K \subset [-\Id,\Id]$. From Proposition \ref{proposition:zonoids}, there is a vector measure
$\mu$ defined on a $\sigma$-algebra $\mathcal{F}$ on a set $\Omega$, whose range is $K$. Let $A \in \mathcal{F}$ such that $\mu(A) = -\Id$.
The vector measure $\mathrm{M}$ defined for $B \in \mathcal{F}$ by
\[ \mathrm{M}(B) = \frac{1}{2} \left( \mu(B \setminus A) - \mu( B \cap A) \right)= \frac{1}{2}\left( \mu( B \Delta A) + \Id \right) \]
is a POVM. Indeed, its range, which equals $\frac{1}{2} ( K +\Id)$, lies inside the positive semidefinite cone, and contains $\Id$. We get from
\eqref{eq:K_M-polar} that $K_\mathrm{M}=K$.

Conversely, for any POVM $\mathrm{M}$, formula \eqref{eq:K_M-polar} implies that $\pm \Id \in K \subset [-\Id,\Id]$.
The fact that $K$ is a zonoid follows, using
the general fact that the convex hull of the range of a vector measure is a zonoid (see \cite{Bolker}, Theorem 1.6).

In the case of zonotopes and discrete POVMs, these arguments have more elementary analogues which we do not repeat.
\end{proof}

We can now argue about the uniqueness part in Proposition \ref{proposition:POVM-states}. This is indeed a consequence
of the uniqueness of the measure associated to a zonoid in Proposition \ref{proposition:zonoids}: after rescaling and symmetrization, a measure
$\mu$ on $S(\C^d)$ satisfying
\eqref{eq:support-function-POVM} naturally induces a measure $\nu$ on the Hilbert--Schmidt sphere satisfying \eqref{eq:support-function-zonoid}
for $K=K_{\mathrm{M}}$.

Another characterization of zonoids involves the Banach space $L^1=L^1([0,1])$. A symmetric convex body $K$
is a zonoid if and only if the normed space $(\R^n,h_K)$ embeds isometrically into $L^1$. Therefore, Proposition \ref{proposition:POVM} can be restated
as a characterization of distinguishability norms on $\mathcal{H}(\C^d)$.

\begin{corollary}
Let $\|\cdot\|$ be a norm on $\mathcal{H}(\C^d)$. The following are equivalent
\begin{enumerate}
 \item There is POVM $\mathrm{M}$ on $\C^d$ such that $\|\cdot\|=\|\cdot\|_{\mathrm{M}}$.
 \item The normed space $(\mathcal{H}(\C^d),\|\cdot\|)$ is isometric to a subspace of $L^1$, and the following inequality is satisfied for any
 $\Delta \in \mathcal{H}(\C^d)$
 \[ | \tr \Delta | \leq \|\Delta\| \leq \tr | \Delta | .\]
\end{enumerate}
\end{corollary}

\section{Local POVMs and tensor products of zonoids}

\label{section:tensorizing}

\subsection{Tensor products for zonoids}

There is a natural notion of tensor product for subspaces of $L^1$ which appeared in the Banach space literature (see e.g. \cite{FJ}).

\begin{definition} \label{definition:1-tensor}
Let $X,Y$ be two Banach spaces which can be embedded isometrically into $L^1$, i.e. such that there exist linear
norm-preserving maps $i : X \to L^1(\mu)$ and $j : Y \to L^1(\nu)$. Then, the $1$-tensor product of $X$ and $Y$ is defined as the completion of the
algebraic
tensor product $X \otimes Y$ for the norm
\[ \left\| \sum_k x_k \otimes y_k \right\|_{X \otimes^1 Y} = \int \int \left| \sum_k i(x_k)(s)j(y_k)(t) \right| \, \mathrm{d}\mu(s)
\mathrm{d}\nu(t) .\]
\end{definition}

It can be checked that the norm above is well-defined and does not depend on the particular choice of the embeddings $i,j$ (see e.g.
\cite{FJ} or Lemma 2 in \cite{RosenthalSzarek}).

In the finite-dimensional case, subspaces of $L^1$ are connected to zonoids. Therefore, Definition \ref{definition:1-tensor}
leads naturally to a notion of tensor product for (symmetric) zonoids.

\begin{definition} \label{def:otimes_Z}
Let $K \subset\R^m$ and $L \subset\R^n$ be two symmetric zonoids, and suppose that $\nu_K$ and $\nu_L$ are Borel measures on $S^{m-1}$ and $S^{n-1}$
respectively, such that for any $x \in \R^m$ and $y \in \R^n$,
\[ h_K(x) = \int_{S^{m-1}} | \langle x,\theta \rangle | \, \mathrm{d}\nu_K(\theta)\ \text{ and }\ h_L(y) = \int_{S^{n-1}}
| \langle y,\phi \rangle | \, \mathrm{d}\nu_L(\phi). \]
The zonoid tensor product of $K$ and $L$ is defined as the zonoid $K \otimes^Z L \subset \R^n \otimes \R^m$ whose support function satisfies
\begin{equation} \label{eq:tensorizing-zonoids} h_{K \otimes^Z L}(z) =\int_{S^{m-1}}\int_{S^{n-1}}
| \langle z,\theta\otimes\phi \rangle | \, \mathrm{d}\nu_K(\theta)\mathrm{d}\nu_L(\phi) \end{equation}
for any $z \in \R^m \otimes \R^n$.
\end{definition}

As in Definition \ref{definition:1-tensor}, this construction does not depend on the choice of the measures $\nu_K$ and $\nu_L$. This can be seen
directly: given $z \in \R^m \otimes \R^n$ and $\phi\in S^{n-1}$, set $\widetilde{z}(\phi)=\left(\Id\otimes\bra{\phi}\right)(z)$. We have
\begin{equation} \label{eq:ztilde-phi} h_{K \otimes^Z L}(z) =
\int_{S^{n-1}}h_K(\widetilde{z}(\phi))\, \mathrm{d}\nu_L(\phi),
\end{equation}
and therefore $K \otimes^Z L$ does not depend on $\nu_K$. The same argument applies for $\nu_L$.

In the case
of zonotopes, the zonoid tensor product takes a simpler form :
\[ \left( \sum_i \conv \{ \pm v_i \} \right) \otimes^Z \left( \sum_j \conv \{ \pm w_j \} \right)
= \sum_{i} \sum_{j} \conv \{ \pm v_i \otimes w_j \} .\]

Here is a first simple property of the zonoid tensor product.

\begin{lemma}
Given symmetric zonoids $K,L$ and linear maps $S,T$, we have
\[ S(K) \otimes^Z T(L) = (S \otimes T) ( K \otimes^Z L) \]
\end{lemma}

Additionally, and crucially for the applications we have in mind, the zonoid tensor product is compatible with inclusions.

\begin{lemma}
\label{lemma:tensorizing-zonoids}
Let $K,K'$ be two symmetric zonoids in $\R^m$ with $K\subset K'$, and let $L,L'$ be two symmetric zonoids in $\R^n$ with $L\subset L'$. Then
\[ K \otimes^Z L \subset K' \otimes^Z L'. \]
\end{lemma}

\begin{proof}
This is a special case of Lemma 2 in \cite{RosenthalSzarek}. Here is a proof in the language of zonoids.
We may assume that $L=L'$, the general case following then by arguing that
$K \otimes^Z L \subset K' \otimes^Z L \subset K' \otimes^Z L'$.

In terms of support functions, we are thus reduced to showing that the inequality $h_K \leq h_{K'}$ implies the inequality
$h_{K \otimes^Z L} \leq h_{K' \otimes^Z L}$,
which is an easy consequence of \eqref{eq:ztilde-phi}.
\end{proof}

Suppose that $(X,\|\cdot\|_X)$ and $(Y,\|\cdot\|_Y)$ are Banach spaces with Euclidean norms, i.e. induced by some inner
products $\scalar{\cdot}{\cdot}_X$ and $\scalar{\cdot}{\cdot}_Y$. Their Euclidean tensor product $X \otimes^2 Y$ is defined (after completion)
by the norm induced by the inner product on the algebraic tensor product which satisfies
\[ \scalar{x \otimes y}{x' \otimes y'} = \scalar{x}{x'}_X \scalar{y}{y'}_Y. \]
It turns out that, for Euclidean norms, the tensor norms $\otimes^1$ and $\otimes^2$ are equivalent.

\begin{proposition}[see \cite{RosenthalSzarek,Bennett}] \label{proposition:2-vs-1}
If $X$ and $Y$ are two Banach spaces equipped with Euclidean norms, then
\[ \sqrt{\frac{2}{\pi}} \| \cdot \|_{X \otimes^2 Y} \leq \| \cdot \|_{X \otimes^1 Y}
\leq \| \cdot \|_{X \otimes^2 Y} .\]
\end{proposition}

\subsection{Local POVMs}

In quantum mechanics, when a system is shared by several parties, the underlying global Hilbert space is the tensor product of the local Hilbert spaces
corresponding to each of the subsystems. A physically relevant class of POVMs on such a multipartite system is the one of local POVMs,
describing the situation where each party is only able to perform measurements on his own subsystem.

\begin{definition}
For $i=1,2$, let $\mathrm{M}_i$ denote a POVM on $\C^{d_i}$, defined on a $\sigma$-algebra $\mathcal{F}_i$ on a set $\Omega_i$. The tensor POVM
$\mathrm{M}_1 \otimes \mathrm{M}_2$ is the unique map defined on the product $\sigma$-algebra $\mathcal{F}_1 \otimes \mathcal{F}_2$ on $\Omega_1
\times \Omega_2$, and such that
\[ (\mathrm{M}_1 \otimes \mathrm{M}_2) ( A_1 \times A_2) = \mathrm{M}_1(A_1) \otimes \mathrm{M}_2(A_2) \]
for every $A_1 \in \mathcal{F}_1, A_2 \in \mathcal{F}_2$. By construction, $\mathrm{M}_1 \otimes \mathrm{M}_2$ is a POVM on $\C^{d_1} \otimes \C^{d_2}$.
\end{definition}

In the discrete case, this definition becomes more transparent: if $\mathrm{M}=(M_i)_{1\leq i\leq m}$ and
$\mathrm{N}=(N_j)_{1\leq j\leq n}$ are discrete POVMs, then $\mathrm{M}\otimes\mathrm{N}$ is also discrete, and
\[ \mathrm{M}\otimes\mathrm{N} = (M_i\otimes N_j)_{1\leq i\leq m,1\leq j\leq n} .\]

POVMs on $\C^{d_1} \otimes \C^{d_2}$ which can be decomposed as tensor product of two POVMs are called local POVMs.
If we identify the POVMs $\mathrm{M}_1$ and $\mathrm{M}_2$ with measures $\mu_1$ and $\mu_2$ as in Proposition \ref{proposition:POVM-states},
then the measure corresponding to $\mathrm{M}_1 \otimes \mathrm{M}_2$ is the image of the product measure $\mu_1 \times \mu_2$ under the map
$(\rho,\sigma) \mapsto \rho \otimes \sigma$. It thus follows that

\begin{proposition} \label{proposition:tensor-POVM}
If $\mathrm{M}$ and $\mathrm{N}$ are two POVMs, then
$ \|\cdot\|_{\mathrm{M} \otimes \mathrm{N}} = \|\cdot\|_{\mathrm{M}} \otimes^1 \|\cdot\|_{\mathrm{N}} $ and
$ K_{\mathrm{M} \otimes \mathrm{N}} = K_{\mathrm{M}} \otimes^Z K_{\mathrm{N}} $.
\end{proposition}

These definitions and statements are given here only in the bipartite case for the sake of clarity, but
can be extended to the situation where a system is shared between any number $k$ of parties.

\section{Sparsifying POVMs}

\label{section:sparsification}

\subsection{The uniform POVM}

It has been proved in \cite{MWW} that, in several senses, the ``most efficient'' POVM on $\C^d$ is the ``most symmetric'' one, i.e.
the uniform POVM $\mathrm{U}_d$, which corresponds to the uniform measure on the set of pure states in the representation
\eqref{eq:support-function-POVM} from Proposition \ref{proposition:POVM}.

The corresponding norm is
\begin{equation} \label{eq:def-normU} \|\Delta\|_{\mathrm{U}_d} = d \E | \langle \psi | \Delta | \psi \rangle |, \end{equation}
where $\psi$ is a random Haar-distributed unit vector.

An important property is that the norm $\|\cdot\|_{\mathrm{U}_d}$ is equivalent
to a ``modified'' Hilbert--Schmidt norm.
\begin{proposition}[\cite{HMS,LW}] \label{proposition:norm-equivalence}
For every $\Delta \in \cH(\C^d)$, we have
\begin{equation} \label{eq:U-vs-modifiedHS} \frac{1}{\sqrt{18}}\|\Delta\|_{2(1)}\leq\|\Delta\|_{\mathrm{U}_d} \leq\|\Delta\|_{2(1)}, \end{equation}
where the norm $\|\cdot\|_{2(1)}$ is defined as
\begin{equation} \label{eq:modifiedHS} \|\Delta\|_{2(1)}=\sqrt{\mathrm{Tr}(\Delta^2)+(\mathrm{Tr}\Delta)^2}. \end{equation}
\end{proposition}

One can check that $\|\Delta\|_{2(1)}$ equals the $L^2$ norm of the random variable $\langle g | \Delta | g \rangle$, where $g$ is a standard Gaussian
vector in $\C^d$, while the $L^1$ norm of this random variable is nothing else than $\|\Delta\|_{\mathrm{U}_d}$. Therefore
Proposition \ref{proposition:norm-equivalence} can be seen as a reverse H\"older inequality, and an interesting problem would be to find the optimal constant in that inequality (the factor $\sqrt{18}$ is presumably far from optimal).

This dimension-free lower bound on the
distinguishing power of the uniform POVM is of interest in quantum information theory.
One could cite as one of its applications the possibility to establish lower-bounds on the dimensionality reduction of quantum states \cite{HMS}. However,
from a computational or algorithmic point of view, this statement involving a continuous POVM is of no practical use.
There has been interest therefore in the question of sparsifying $\mathrm{U}_d$, i.e. of finding a discrete POVM, with as few outcomes as possible,
which would be equivalent to $\mathrm{U}_d$ in terms of discriminating efficiency. Examples of such constructions arise from the theory
of projective $4$-designs.

Given an integer $t \geq 1$, an (exact) $t$-design is a finitely supported probability measure $\mu$ on $S_{\C^d}$ such that
\[ \int_{S_{\C^d}} \ketbra{\psi}{\psi}^{\otimes t} \, \mathrm{d}\mu(\psi) =
\int_{S_{\C^d}} \ketbra{\psi}{\psi}^{\otimes t} \, \mathrm{d}\sigma(\psi) = \binom{d+t-1}{t}^{-1} P_{\Sym^t(\C^d)} .\]
Here, $\sigma$ denotes the Haar probability measure on $S_{\C^d}$, and $P_{\Sym^t(\C^d)}$ denotes the orthogonal projection onto the symmetric subspace
$\Sym^t(\C^d) \subset (\C^d)^{\otimes t}$.

Note that a $t$-design is also a $t'$-design for any $t' \leq t$. Let $\mu$ be a $1$-design. The map $\psi \mapsto \ketbra{\psi}{\psi}$
pushes forward $\mu$ into a measure $\tilde{\mu}$ on the set of (pure) states, with barycenter equal to $\Id/d$. By Proposition \ref{proposition:POVM},
this measure corresponds to a POVM, and in the following we identify $t$-designs with the associated POVMs. For example the uniform POVM
$\mathrm{U}_d$ is a $t$-design for any $t$.

This notion can be relaxed: define an $\e$-approximate $t$-design to be a finitely supported measure $\mu$ on $S_{\C^d}$ such that
\[ (1-\e) \int_{S_{\C^d}} \ketbra{\psi}{\psi}^{\otimes t} \, \mathrm{d}\sigma(\psi) \leq
\int_{S_{\C^d}} \ketbra{\psi}{\psi}^{\otimes t} \, \mathrm{d}\mu(\psi) \leq
(1+\e) \int_{S_{\C^d}} \ketbra{\psi}{\psi}^{\otimes t} \, \mathrm{d} \sigma(\psi) .\]

It has been proved in \cite{AE} that a $4$-design (exact or approximate) supported on $N$ points yields a POVM $\mathrm{M}$ with $N$ outcomes such that
\begin{equation} \label{eq-AE} C^{-1} \|\cdot\|_{\mathrm{U}_d} \leq \|\cdot\|_{\mathrm{M}} \leq C \|\cdot\|_{\mathrm{U}_d} \end{equation}
for some constant $C$. The proof is based on the fourth moment method, which is used to control the first absolute moment of a random variable
by its second and fourth moments.

Now, what is the minimal cardinality of a $4$-design? The support of any exact or $\e$-approximate (provided $\e<1$) $4$-design must
contain at least $\dim (\Sym^4(\C^d)) = \binom{d+3}{4} = \Omega(d^4)$ points. Conversely, an argument
based on Carath\'eodory's theorem shows that there exist exact $4$-designs with $O(d^8)$ points. Starting from such an exact $4$-design,
the sparsification procedure from \cite{BSS} gives a deterministic and efficient algorithm which outputs an $\e$-approximate $4$-design supported
by $O(d^4/\e^2)$ points.

However, this approach has two drawbacks: the constant $C$ from \eqref{eq-AE} cannot be taken close to $1$, and the number of outcomes
has to be $\Omega(d^4)$. We are going to remove both inconveniences in our Theorem \ref{theorem:approximation-of-U}.

\subsection{Euclidean subspaces}

How do these ideas translate into the framework of zonoids? The analogue of $\mathrm{U_d}$ is the most symmetric zonoid, namely the Euclidean
ball $B_2^n \subset \R^n$. To connect with literature from functional analysis, it is worth emphasizing that approximating $B_2^n$ by a zonotope
which is the sum of $N$ segments is equivalent to embedding the space $\ell_2^n=(\R^n,\|\cdot\|_2)$ into the space $\ell_1^N=(\R^N,\|\cdot\|_1)$.
Indeed, assume that $x_1,\dots,x_N$ are points in $\R^n$ such that, for some constants $c,C$,
\[ c Z \subset B_2^n \subset C Z ,\]
where $Z = \conv \{\pm x_1 \} + \dots + \conv \{ \pm x_N \}$. Then the map $u: \R^n \to \R^N$ defined by
\[ u(x) = \Big( \scalar{x}{x_1}, \cdots, \scalar{x}{x_N} \Big) \]
satisfies $c \|u(x)\|_1 \leq \|x\|_2 \leq C\|u(x)\|_1$ for any $x \in \R^n$. In this context,
the ratio $C/c$ is often called the distortion of the embedding.

An early result by Rudin \cite{Rudin60} shows an explicit embedding of $\ell_2^n$ into $\ell_1^{O(n^2)}$ with distortion $\sqrt{3}$. This is proved by
the fourth moment method and can be seen as the analogue of the constructions based on 4-designs.
The following theorem (a variation on Dvoretzky's theorem)
has been a major improvement on Rudin's result, showing that $\ell_1^N$ has almost Euclidean sections of proportional dimension.

\begin{theorem}[\cite{FLM}] \label{theorem-dvoretzky}
For every $0<\e<1$, there exists a subspace $E \subset \R^N$ of dimension $n=c(\e)N$ such that for any $x \in E$,
\begin{equation} \label{eq:dvoretzky} (1-\e) M \|x\|_2 \leq \|x\|_1 \leq (1+\e) M\|x\|_2, \end{equation}
where $M$ denotes the average of the $1$-norm over the Euclidean unit sphere $S^{N-1}$.
\end{theorem}

Theorem \ref{theorem-dvoretzky} was first proved in \cite{FLM}, making a seminal use of measure concentration in the form of L\'evy's lemma. The argument
shows that a generic subspace $E$ (i.e. picked uniformly at random amongst all $c(\e)N$-dimensional subspaces of $\R^N$) satisfies the conclusion of the
theorem with high probability for $c(\e)=O\left(\e^2 |\log \e|^{-1}\right)$.
This was later improved in \cite{Gordon} to $c(\e)=O\left(\e^2\right)$.

\subsection{Sparsification of the uniform POVM}

Translated in the language of zonotopes, Theorem \ref{theorem-dvoretzky} states that the sum of $O(n)$ randomly chosen segments in $\R^n$
is close to the Euclidean ball $B_2^n$. More precisely, for any $0<\e<1$, the zonotope $Z=\conv\{\pm x_1\}+\cdots+\conv\{\pm x_N\}$,
with $N=c(\e)^{-1}n$ and $x_1,\ldots,x_N$ randomly chosen points in $\R^n$, is $\e$-close to the Euclidean ball $B_2^n$, in the sense
that $(1-\e)Z\subset B_2^n\subset(1+\e)Z$.

By analogy, we expect a POVM constructed from $O(d^2)$ randomly chosen elements to be close to the uniform POVM.
This random construction can be achieved as follows:
let $(\ket{\psi_i})_{1\leq i \leq n}$ be independent random vectors, uniformly chosen on the unit sphere of $\C^d$. Set
$P_i = \ketbra{\psi_i}{\psi_i}$, $1 \leq i \leq n$, and $S = P_1 + \dots + P_n$. When $n \geq d$, $S$ is almost surely invertible,
and we may consider the random POVM
\begin{equation} \label{eq:randomPOVM} \mathrm{M} = (S^{-1/2} P_i S^{-1/2})_{1 \leq i \leq n } .\end{equation}

\begin{theorem} \label{theorem:approximation-of-U}
Let $\mathrm{M}$ be a random POVM on $\C^d$ with $n$ outcomes, defined as in \eqref{eq:randomPOVM}, and let $0<\e<1$.
If $n \geq C\e^{-2} |\log \e| d^2$, then with high probability
the POVM $\mathrm{M}$ satisfies the inequalities
\[ (1-\e) \|\Delta\|_{\mathrm{U}_d} \leq \|\Delta\|_{\mathrm{M}} \leq (1+\e) \|\Delta\|_{\mathrm{U}_d} \]
for every $\Delta \in \cH(\C^d)$.
\end{theorem}

By ``with high probability'' we mean that the probability that the conclusion fails is less than $\exp(-c(\e)d)$ for some constant $c(\e)$.
Theorem \ref{theorem:approximation-of-U} is proved in Section \ref{sec:uniform-POVM}, the proof being based on a careful use of $\e$-nets and deviation
inequalities. It does not seem possible to deduce formally Theorem \ref{theorem:approximation-of-U} from the existing Banach space literature.

Theorem \ref{theorem:approximation-of-U} shows that the uniform POVM on $\C^d$ can be $\e$-approximated
(in the sense of closeness of distinguishability norms)
by a
POVM with $n=O(\e^{-2}|\log\e|d^2)$ outcomes. Note that the dependence of $n$ with respect to $d$ is optimal: since a POVM on $\C^d$ must have at least
$d^2$ outcomes to be informationally complete, one cannot hope for a tighter dimensional dependence. The dependence with respect to $\e$ is less clear:
the factor $|\log \e|$ can probably be removed but we do not pursue this direction.

Our construction is random and a natural question is whether deterministic constructions yielding comparable properties exist. A lot of effort has
been put in derandomizing Theorem \ref{theorem-dvoretzky}. We refer to \cite{IS} for bibliography and mention two of the latest results.
Given any $0<\gamma<1$, it is shown in \cite{IS} how to construct, from $cn^{\gamma}$ random bits (i.e. an amount of randomness sub-linear in $n$) a
subspace of $\ell_1^N$ satisfying
\eqref{eq:dvoretzky} with $N\leq (\gamma\e)^{-C\gamma}n$. A completely explicit construction appears in \cite{Indyk}, with
$N\leq n2^{C(\e)(\log\log n)^2}=n^{1+C(\e)o(n)}$. It is not obvious how to adapt these constructions to obtain sparsifications of the uniform POVM
using few or no randomness.

\subsection{Sparsification of any POVM}

Theorem \ref{theorem-dvoretzky}
initiated intensive research in the late 80's \cite{Schechtman,BLM,Talagrand} on the theme of ``approximation of zonoids by zonotopes'',
trying to extend the result for
the Euclidean ball (the most symmetric zonoid) to an arbitrary zonoid. This culminated in Talagrand's proof \cite{Talagrand} that for any zonoid
$Y \subset \R^n$ and any $0<\e<1$, there exists a zonotope $Z \subset \R^n$ which is the sum of $O(\e^{-2}n \log n)$ segments and such that
$(1-\e)Y \subset Z \subset (1+\e)Y$. A more precise version is stated in Section \ref{sec:approximation-any}. Whether the $\log n$ factor
can be removed is still an open problem.

This result easily implies a similar result for POVMs, provided we consider the larger class of sub-POVMs.
A discrete sub-POVM with $n$ outcomes is a
finite family $\mathrm{M} = (M_i)_{1\leq i\leq n}$ of $n$ positive operators such that $S = \sum_{i=1}^n M_i \leq \Id$.
As for POVMs, the norm associated to a sub-POVM $\mathrm{M}$ is defined for $\Delta \in \cH(\C^d)$ by
\[ \| \Delta \|_{\mathrm{M}} = \sum_{i=1}^n |\tr (\Delta M_i)|. \]
We prove the following result in Section \ref{sec:approximation-any}.

\begin{theorem}
\label{theorem:approximation-any}
Given any POVM $\mathrm{M}$ on $\C^d$ and any $0<\e<1$, there is a sub-POVM $\mathrm{M}' = (M'_i)_{1\leq i\leq n}$, with $n \leq C \e^{-2}d^2 \log(d)$
such that, for any $\Delta \in \cH(\C^d)$,
\[ (1-\e) \| \Delta \|_{\mathrm{M}} \leq \|\Delta\|_{\mathrm{M}'} \leq \|\Delta\|_{\mathrm{M}}. \]
Moreover, we can guarantee that the states $M'_i/\tr (M'_i)$ belong to the support of the measure $\mu$ associated to $\mathrm{M}$.
\end{theorem}

We do not know whether Theorem \ref{theorem:approximation-any} still holds if we want $\mathrm{M}'$ to be a POVM.
Given a sub-POVM $(M_i)_{1\leq i\leq n}$, there are at least
two natural ways to modify it into a POVM. A solution is to add an extra outcome corresponding to the operator $\Id-S$,
and another one is to substitute $S^{-1/2}M_iS^{-1/2}$
in place of $M_i$, as we proceeded in \eqref{eq:randomPOVM}. However for a general POVM, the error terms arising from this renormalization step may
exceed the quantity to be approximated.

\section{Sparsifying local POVMs}
\label{section:sparsification-multipartite}

Proposition \ref{prop:tensorizing-sparsifications} below is an immediate corollary of Lemma \ref{lemma:tensorizing-zonoids} and Proposition
\ref{proposition:tensor-POVM}. In words, it shows that, on a
multipartite system, a local POVM can be sparsified by tensorizing sparsifications of each of its factors.

\begin{proposition}
\label{prop:tensorizing-sparsifications}
Let $0<\e<1$. Let $\mathrm{M}_1,\ldots,\mathrm{M}_k$ be POVMs and $\mathrm{M}_1',\ldots,\mathrm{M}_k'$ be (sub-)POVMs, on $\C^{d_1},\ldots,\C^{d_k}$
respectively, satisfying, for all $1\leq i\leq k$, and for all $\Delta\in\cH(\C^{d_i})$,
\[ (1-\e)\|\Delta\|_{\mathrm{M}_i}\leq\|\Delta\|_{\mathrm{M}_i'}\leq(1+\e)\|\Delta\|_{\mathrm{M}_i}. \]
Then, for any $\Delta\in\cH(\C^{d_1}\otimes\cdots\otimes\C^{d_k})$,
\[ (1-\e)^k\|\Delta\|_{\mathrm{M}_1\otimes\cdots\otimes\mathrm{M}_k} \leq\|\Delta\|_{\mathrm{M}_1'\otimes\cdots\otimes\mathrm{M}_k'}
\leq(1+\e)^k\|\Delta\|_{\mathrm{M}_1\otimes\cdots\otimes\mathrm{M}_k}. \]
\end{proposition}


Let us give a concrete application of Proposition \ref{prop:tensorizing-sparsifications}. We consider $k$ finite-dimensional Hilbert spaces
$\C^{d_1},\ldots,\C^{d_k}$ and define the local uniform POVM on the $k$-partite Hilbert space $\C^{d_1}\otimes\cdots\otimes\C^{d_k}$ as the tensor
product of the $k$ uniform POVMs $\mathrm{U}_{d_1},\ldots,\mathrm{U}_{d_k}$. We will denote it by $\mathrm{LU}$. The
corresponding distinguishability norm can be described, for any $\Delta\in\cH(\C^{d_1}\otimes\cdots\otimes\C^{d_k})$, as
\[ \|\Delta\|_{\mathrm{LU}} = d \E \left| \langle \psi_1 \otimes \cdots \otimes \psi_k | \Delta | \psi_1 \otimes \cdots \otimes \psi_k
\rangle \right|, \]
where $d=d_1\times\cdots\times d_k$ is the dimension of the global Hilbert space, and where the random unit vectors $\psi_1, \dots ,\psi_k$ are
independent and Haar-distributed in $\C^{d_1},\ldots,\C^{d_k}$ respectively.

The following multipartite generalization of Proposition \ref{proposition:norm-equivalence} shows that the norm $\|\cdot\|_{\mathrm{LU}}$,
in analogy to the norm $\|\cdot\|_{\mathrm{U}}$, is equivalent to a ``modified'' Hilbert--Schmidt norm.

\begin{proposition} [\cite{LW}] \label{proposition:norm-equivalence-multipartite}
For every $\Delta\in\cH(\C^{d_1}\otimes\cdots\otimes\C^{d_k})$, we have
\begin{equation} \label{eq:U-vs-modifiedHS-multi} \frac{1}{18^{k/2}}\|\Delta\|_{2(k)}\leq\|\Delta\|_{\mathrm{LU}} \leq\|\Delta\|_{2(k)}, \end{equation}
where the norm $\|\cdot\|_{2(k)}$ is defined as
\begin{equation} \label{eq:modifiedHS-multi} \|\Delta\|_{2(k)}=\sqrt{\sum_{I\subset \{1,\dots,k\}} \tr \left[\big(\tr_{I}\Delta\big)^2\right]}.
\end{equation}
Here $\tr_I$ denotes the partial trace over all parties $I \subset \{1,\dots,k\}$.
\end{proposition}

\begin{proof}[Proof of Proposition \ref{proposition:norm-equivalence-multipartite}]
A direct proof appears in \cite{LW}, but we find interesting to show that in can be deduced (with a worst constant) from Proposition
\ref{proposition:norm-equivalence}.
If we denote by $\scalar{\cdot}{\cdot}_H$ the inner product inducing a Euclidean norm $\|\cdot\|_H$, we have
\[ \langle A_1\otimes\cdots\otimes A_k, B_1\otimes\cdots\otimes B_k \rangle_{2(k)} =\langle A_1, B_1 \rangle_{2(1)}\times\cdots\times\langle A_k,
B_k \rangle_{2(1)} \]
which is equivalent to saying that
\[ \|\cdot\|_{2(k)} = \|\cdot\|_{2(1)} \otimes^2 \cdots \otimes^2 \|\cdot\|_{2(1)} .\]
We thus get by Proposition \ref{proposition:2-vs-1},
\[ c_0^{k-1}\|\cdot\|_{2(k)} \leq \|\cdot\|_{2(1)} \otimes^1 \cdots
\otimes^1  \|\cdot\|_{2(1)} \leq \|\cdot\|_{2(k)} \]
with $c_0=\sqrt{2/\pi}$.
Now, we also know by Proposition \ref{proposition:tensor-POVM} that on $\cH(\C^{d_1}\otimes\cdots\otimes\C^{d_k})$, $\|\cdot\|_{\mathrm{LU}}=
\|\cdot\|_{\mathrm{U}_{d_1}}\otimes^1\cdots\otimes^1\|\cdot\|_{\mathrm{U}_{d_k}}$, and by Proposition \ref{proposition:norm-equivalence}
that $c\|\cdot\|_{2(1)}\leq\|\cdot\|_{\mathrm{U}_d}\leq \|\cdot\|_{2(1)}$ for some constant $c$ ($c=1/\sqrt{18}$ works). So by Lemma
\ref{lemma:tensorizing-zonoids},
\[ c^k \, \|\cdot\|_{2(1)} \otimes^1 \cdots \otimes^1 \|\cdot\|_{2(1)}
\leq \|\cdot\|_{\mathrm{LU}} \leq  \|\cdot\|_{2(1)} \otimes^1 \cdots \otimes^1 \|\cdot\|_{2(1)}, \]
and therefore
\[ c_0^{k-1} c^k \|\cdot\|_{2(k)} \leq \|\cdot\|_{\mathrm{LU}} \leq \|\cdot\|_{2(k)}.  \qedhere\]
\end{proof}

Remarkably, local dimensions do not appear in equation \eqref{eq:U-vs-modifiedHS-multi}.
This striking fact that local POVMs can have asymptotically non-vanishing distinguishing power can be used to construct an algorithm that solves the
Weak Membership Problem for separability in quasi-polynomial time (see \cite{BCY} for a description in the bipartite case). Hence the importance of
being able to sparsify the
local uniform POVM by a POVM for which the locality property is preserved and which has a number of outcomes that optimally scales as the square of
the global dimension. We state the corresponding multipartite version of Theorem \ref{theorem:approximation-of-U}, which is straightforwardly obtained
by combining the unipartite version with Proposition \ref{prop:tensorizing-sparsifications}.

\begin{theorem} \label{theorem:approximation-LU}
Let $0<\e<1$. For all $1\leq i\leq k$, let $\mathrm{M}_i$ be a random POVM on $\C^{d_i}$ with $n_i\geq C\e^{-2}|\log\e|d_i^2$ outcomes, defined as in
\eqref{eq:randomPOVM}. Then, with high probability, the local POVM $\mathrm{M}_1\otimes\cdots\otimes\mathrm{M}_k$ on
$\C^{d_1}\otimes\cdots\otimes\C^{d_k}$ is such that, for any $\Delta\in\cH(\C^{d_1}\otimes\cdots\otimes\C^{d_k})$,
\[ (1-\e)^k\|\Delta\|_{\mathrm{LU}} \leq\|\Delta\|_{\mathrm{M}_1\otimes\cdots\otimes\mathrm{M}_k} \leq(1+\e)^k\|\Delta\|_{\mathrm{LU}}. \]
\end{theorem}

Let us rephrase the content of Theorem \ref{theorem:approximation-LU}: the local uniform POVM on $\C^{d_1}\otimes\cdots\otimes\C^{d_k}$ can be
$k\e$-approximated (in terms of distinguishability norms) by a POVM which is also local and has a total number of outcomes
$n=O(C^k\e^{-2k}|\log\e|^kd^2)$, where $d=d_1\times\cdots\times d_k$. Note that the dimensional dependence of $n$ is optimal. On the contrary,
the dependence of $n$ on $\e$ deteriorates as $k$ grows. The high-dimensional situation our result applies to is thus really
the one of a ``small'' number of ``large'' subsystems (i.e. $k$ fixed and $d_1,\ldots,d_k\rightarrow+\infty$), and not of
a ``large'' number of ``small'' subsystems.

\section{Proof of Theorem \ref{theorem:approximation-of-U}}
\label{sec:uniform-POVM}


In this section we prove Theorem \ref{theorem:approximation-of-U}.
Let $n\in\N$ and $(\ket{\psi_i})_{1 \leq i \leq n}$ be independent random unit vectors,
uniformly distributed on the unit sphere of $\C^d$. Our main technical estimates are a couple of probabilistic inequalities. Proposition
\ref{proposition:Wishart} is an immediate consequence of Theorem 1 in \cite{Aubrun}. Proposition \ref{proposition:large-deviations} is a consequence
of Bernstein inequalities. However, its proof requires some careful estimates which we postpone to Section \ref{section:bernstein}.

\begin{proposition} \label{proposition:Wishart}
If $(\ket{\psi_i})_{1 \leq i \leq n}$ are independent random vectors, uniformly distributed on the unit sphere of $\C^d$, then for every $0<\eta<1$
\[ \P \left( (1-\eta) \frac{\Id}{d} \leq \frac{1}{n} \sum_{i=1}^n \ketbra{\psi_i}{\psi_i} \leq (1+\eta) \frac{\Id}{d} \right) \geq 1 -
C^d \exp(-c n \eta^2). \]
\end{proposition}

\begin{proposition} \label{proposition:large-deviations}
Let $\Delta \in \cH(\C^d)$, and $(\ket{\psi_i})_{1 \leq i \leq n}$ be independent random vectors, uniformly distributed on the unit sphere of $\C^d$.
For $1 \leq i \leq n$, consider the random variables $X_i = d| \bra{\psi_i} \Delta \ket{\psi_i} |$ and $Y_i=X_i-\E X_i=X_i - \|\Delta\|_{\mathrm{U}_d}$.
Then, for any $t > 0$,
\[ \P \left( \left| \frac{1}{n} \sum_{i=1}^n Y_i \right| \geq t \| \Delta \|_{\mathrm{U}_d} \right) \leq 2 \exp ( - c'_0 n \min(t,t^2)) .\]
\end{proposition}

We now show how to derive Theorem \ref{theorem:approximation-of-U} from the estimates in Propositions \ref{proposition:Wishart} and \ref{proposition:large-deviations}.
For each $1\leq i\leq n$, set $P_i=\ketbra{\psi_i}{\psi_i}$, and introduce the (random) norm defined for any $\Delta\in\cH(\C^d)$ as
\[ |||\Delta||| = \frac{d}{n} \sum_{i=1}^n |\tr (\Delta P_i)|. \]
We will now prove that $|||\cdot|||$ is, with probability close to $1$, a good approximation to $\|\cdot\|_{\mathrm{U}_d}$.
First, using Proposition \ref{proposition:large-deviations}, we obtain that for any $0 < \e < 1$ and any $\Delta \in \cH(\C^d)$
\begin{equation} \label{eq:single-Delta}
\P \left( (1-\e) \|\Delta\|_{\mathrm{U}_d} \leq |||\Delta||| \leq (1+\e) \|\Delta\|_{\mathrm{U}_d} \right) \geq 1- 2\exp(-c'_0 n \e^2).
\end{equation}

We next use a net argument. Fix $0<\e<1/3$ and a $\e$-net $\mathcal{N}$ inside the unit ball for the norm $\|\cdot\|_{\mathrm{U}_d}$,
with respect to the distance induced by
$\|\cdot\|_{\mathrm{U}_d}$. A standard volumetric argument (see \cite{Pisier}, Lemma 4.10) shows that we may assume
$\card (\mathcal{N}) \leq (1 + 2/\e)^{d^2} \leq (3/\e)^{d^2}$.
Introduce the quantities
\[ A := \sup \{ ||| \Delta ||| \st \| \Delta \|_{\mathrm{U}_d} \leq 1 \}, \]
\[ A' := \sup \{ ||| \Delta ||| \st \Delta \in \mathcal{N} \}. \]
Given $\Delta$ such that $\|\Delta\|_{\mathrm{U}_d} \leq 1$, there is $\Delta_0 \in \mathcal{N}$ with $\|\Delta-\Delta_0\|_{\mathrm{U}_d} \leq \e$.
By the triangle inequality,
we have $|||\Delta||| \leq A' + |||\Delta - \Delta_0||| \leq A'+ \e A$. Taking supremum over $\Delta$ yields $A \leq A'+\e A$ i.e. $A \leq
\frac{A'}{1-\e}$.

If we introduce $B := \inf \{ ||| \Delta ||| \st \| \Delta \|_{\mathrm{U}_d} = 1 \}$ and $B' := \inf \{ ||| \Delta ||| \st \Delta \in \mathcal{N} \}$,
a similar argument shows that $B \geq B' - \e A$, so that in fact $B \geq B' - \frac{\e A'}{1-\e}$.
We therefore have the implications
\begin{equation} \label{eq:implication} 1-\e \leq B' \leq A' \leq 1+\e \ \ \Longrightarrow \ \  1-\e-\frac{\e(1+\e)}{1-\e} \leq  B \leq A \leq
\frac{1+\e}{1-\e}
\ \ \Longrightarrow \ \ 1-3\e \leq B \leq A \leq 1+ 3\e. \end{equation}
By the union bound, we get from \eqref{eq:single-Delta} that $\P(1-\e \leq B' \leq A' \leq 1+\e) \geq 1-2 \card( \mathcal{N}) \exp(-c'_0n \e^2)$.
Combined with (\ref{eq:implication}), and using homogeneity of norms, this yields
\begin{equation} \label{eq:global-Delta} \P \Big( (1-3\e) \|\cdot \|_{\mathrm{U}_d} \leq |||\cdot||| \leq (1+3\e) \|\cdot\|_{\mathrm{U}_d} \Big) \geq 1- 2 \left(\frac{3}{\e}\right)^{d^2}
\exp(-c'_0n \e^2). \end{equation}
This probability estimate is non-trivial, and can be made close to $1$, provided $n \gtrsim d^2 \e^{-2} | \log \e |$.

Whenever $n \geq d$, the
vectors $(\ket{\psi_i})_{1 \leq i \leq n}$ generically span $\C^d$, and therefore the operator $S = P_1 + \cdots + P_n$ is invertible. We may
then define $\widetilde{P}_i = S^{-1/2} P_i S^{-1/2}$ so that $\mathrm{M} = (\widetilde{P}_i)_{1 \leq i \leq n}$ is a POVM. The norm associated to
$\mathrm{M}$ is, for any $\Delta\in\cH(\C^d)$,
\[ \|\Delta\|_{\mathrm{M}} = \sum_{i=1}^n |\tr (\Delta\widetilde{P}_i)|. \]
We now argue that the norms $|||\cdot|||$ and
$\|\cdot\|_{\mathrm{M}}$ are similar enough (modulo normalization), because the modified operators $\widetilde{P}_i$ are close
enough to the initial ones $P_i$. This is achieved by showing that $T:=\left(\frac{d}{n} S \right)^{-1/2}$ is close to $\Id$ (in operator-norm distance).
We use Proposition \ref{proposition:Wishart} for $\eta = \e \|\Delta\|_{\mathrm{U}_d}/\|\Delta\|_1$.
By Proposition \ref{proposition:norm-equivalence}, we have $\eta \geq \e / \sqrt{18d}$.
Proposition \ref{proposition:Wishart} implies that
\begin{equation} \label{eq:boundZ} \P( \|T-\Id\|_{\infty} \geq \eta ) \leq \P( \|T^{-2}-\Id\|_{\infty} \geq \eta ) \leq C^d \exp(-c' n \e^2 /d).
\end{equation}
This upper bound is much smaller than $1$ provided $n \geq C_1 \e^{-2}d^2$. Also, note that the event $\|T-\Id\|_{\infty} \leq \eta$
implies that
\[ \|\Delta-T \Delta T\|_{\mathrm{M}} \leq \|\Delta - T\Delta T\|_{1} \leq  \|\Delta\|_1 \|\Id-T\|_{\infty}
\left( 1 + \|T\|_{\infty} \right) \leq  2 \eta \|\Delta\|_1 = 2 \e \|\Delta\|_{\mathrm{U}_d}. \]

Using the cyclic property of the trace, we check that $\|T\Delta T\|_{\mathrm{M}} = |||\Delta|||$.
Now, choose $n$ larger than both $C_0 \e^{-2} |\log \e| d^2$ and $C_1 \e^{-2}d^2$.
With high probability, the events from equations \eqref{eq:global-Delta} and \eqref{eq:boundZ} both hold.
We then obtain for every $\Delta \in \cH(\C^d)$,
\[
\|\Delta\|_{\mathrm{M}} \leq \| T\Delta T\|_{\mathrm{M}} + \|\Delta-T\Delta T\|_{\mathrm{M}} \leq |||\Delta||| + 2\e \|\Delta\|_{\mathrm{U}_d} \leq
(1+ 5\e) \|\Delta\|_{\mathrm{U}_d}
\]
and similarly $\|\Delta\|_{\mathrm{M}} \geq (1- 5\e) \|\Delta\|_{\mathrm{U}_d}$.
This is precisely the result from Theorem \ref{theorem:approximation-of-U} with $5\e$ instead of $\e$, which of course can be absorbed by renaming
the constants appropriately.

\section{Proof of Proposition \ref{proposition:large-deviations}}

\label{section:bernstein}

The proof is a direct application of a large deviation inequality for sums of
independent sub-exponential (or $\psi_1$) random variables. Recall that the $\psi_1$-norm of a random variable $X$ (which quantifies the exponential
decay
of the tail) may be defined via the growth of even moments
\[ \|X\|_{\psi_1} := \sup_{q \in \N}\frac{1}{2q}\big(\E |X|^{2q}\big)^{1/{2q}} .\]
This definition is more practical than the standard definition through the Orlicz function $x \mapsto \exp(x)-1$, and leads to an equivalent norm
(see \cite{CGLP},
Corollary 1.1.6). The large deviation inequality for a sum of independent $\psi_1$ random variables is known as Bernstein's inequality.

\begin{theorem}[Bernstein's inequality, see \cite{CGLP}, Theorem 1.2.5.]
\label{th:Bernstein}
Let $X_1,\ldots,X_n$ be $n$ independent $\psi_1$ random variables with mean zero.
Setting $M=\underset{1\leq i\leq n}{\max}\|X_i\|_{\psi_1}$ and $\sigma^2=\frac{1}{n}\underset{1\leq i\leq n}{\sum}\|X_i\|_{\psi_1}^2$, we have
\[ \forall\ t>0,\ \P\left(\left|\frac{1}{n}\sum_{i=1}^n X_i\right|\geq t\right)\leq 2\exp\left(-c_0 n\min\left(\frac{t^2}{\sigma^2},\frac{t}{M}\right)
\right),\]
$c_0>0$ being a universal constant.
\end{theorem}

For $\Delta \in \cH(\C^d)$, consider the random variables $X_i = d|\tr (\Delta P_i)|$ with $P_i=\ketbra{\psi_i}{\psi_i}$,
and $Y_i=X_i - \E X_i =d|\tr (\Delta P_i)| - \|\Delta\|_{\mathrm{U}_d}$. The random variables
$Y_i$ are independent and have mean zero. The key lemma is a bound on their $\psi_1$ norm.

\begin{lemma} \label{lemma:psi1}
 Let $\Delta\in \cH(\C^d)$ and consider the random variable $X:=d|\tr (\Delta P) |$, where $P=\ketbra{\psi}{\psi}$,
with $\psi$ uniformly distributed on the unit sphere of $\C^d$.
Then $\|X\|_{\psi_1} \leq \|\Delta\|_{2(1)}$ and $\|X - \E X\|_{\psi_1} \leq 3 \|\Delta\|_{2(1)} \leq 3\sqrt{18} \|\Delta\|_{\mathrm{U}_d}$.
\end{lemma}

Therefore, we may apply Bernstein's inequality with $M= \sigma \leq 3\sqrt{18} \| \Delta \|_{\mathrm{U}_d} $, yielding Proposition \ref{proposition:large-deviations}.

\begin{proof}[Proof of Lemma \ref{lemma:psi1}]
For each integer $q$, we compute
\[ \E \left[\tr (\Delta P)\right]^{2q}   = \E \tr\left(\Delta^{\otimes 2q}P^{\otimes 2q}\right) = \tr \left( \Delta^{\otimes 2q}
\left[\E P^{\otimes 2q}\right] \right) .\]
We use the fact (see e.g. \cite{Harrow}) that
\[ \E P^{\otimes 2q} = \frac{(2q)!}{(d+2q-1)\times\cdots\times d}
P_{\Sym^{2q}(\C^d)} =\frac{1}{(d+2q-1)\times\cdots\times d} \sum_{\pi\in\mathfrak{S}_{2q}} U(\pi) ,\]
where $P_{\Sym^{2q}(\C^d)}$ denotes the orthogonal projection
onto the symmetric subspace $\Sym^{2q}(\C^d)\subset(\C^d)^{\otimes 2q}$, and for each permutation $\pi\in\mathfrak{S}_{2q}$, $U(\pi)$ denotes the
associated permutation
unitary on $(\C^d)^{\otimes 2q}$. This yields
\[ \E \left[\tr (\Delta P)\right]^{2q} = \frac{1}{(d+2q-1)\times\cdots\times d}\sum_{\pi\in\mathfrak{S}_{2q}}
\tr \left(\Delta^{\otimes 2q}U(\pi)\right) .\]
If $\ell_1,\dots,\ell_k$ denote the lengths of the cycles appearing in the cycle decomposition of a permutation $\pi \in \mathfrak{S}_{2q}$, we have
$\ell_1 + \cdots + \ell_k = 2q$ and
\[ \tr \left(\Delta^{\otimes 2q}U(\pi)\right) = \prod_{i=1}^k \tr (\Delta^{\ell_i}). \]
Now, for any integer $\ell \geq 2$, we have $| \tr (\Delta^\ell) | \leq [\tr (\Delta^2)]^{\ell/2} \leq \|\Delta\|_{2(1)}^\ell$. The inequality
$| \tr (\Delta^\ell) | \leq \|\Delta\|_{2(1)}^\ell$ is also (trivially) true for $\ell=1$. Therefore
$\left|\tr \left(\Delta^{\otimes 2q}U(\pi)\right) \right| \leq \|\Delta\|_{2(1)}^{2q}$. It follows that
\[ \E \left[\tr (\Delta P)\right]^{2q} \leq \frac{(2q)!}{d^{2q}} \|\Delta\|_{2(1)}^{2q} \leq \left( \frac{2q \|\Delta\|_{2(1)}}{d} \right)^{2q}, \]
so that $\left( \E X^{2q} \right)^{1/2q} \leq 2q \|\Delta\|_{2(1)}$, and thus $\|X\|_{\psi_1} \leq \|\Delta\|_{2(1)}$.
The last part of the Lemma follows from the triangle
inequality, since $\| \E X \|_{\psi_1}=|\E X| \leq 2 \| X \|_{\psi_1}$, and from the equivalence \eqref{eq:U-vs-modifiedHS} between the norms
$\|\cdot\|_{\mathrm{U}_d}$
and $\|\cdot\|_{2(1)}$.
\end{proof}

\section{Proof of Theorem \ref{theorem:approximation-any}}
\label{sec:approximation-any}

Here is a version of Talagrand's theorem which is suitable for our purposes.

\begin{theorem}[\cite{Talagrand}] \label{theorem:talagrand}
Let $Z \subset \R^n$ be a symmetric zonotope, with
\[ Z = \sum_{i \in I} \conv \{ \pm u_i \} \]
for a finite family of vectors $(u_i)_{i \in I}$. Then for every $\e >0$ there exists a subset $J \subset I$
with $\card J \leq C n \log n/\e^2$, and positive numbers $(\lambda_i)_{i \in J}$ such that the zonotope
\[ Z' = \sum_{i \in J} \conv \{ \pm \lambda_i u_i \} \]
satisfies $Z' \subset Z \subset (1+\e) Z'$.
\end{theorem}

Theorem \ref{theorem:approximation-any} is a very simple consequence of Theorem \ref{theorem:talagrand}. Let $\mathrm{M}$ be a POVM to
be sparsified. Using Corollary
\ref{corollary:approximation}, we may assume that $\mathrm{M} = (M_i)_{i \in I}$ is discrete. Applying Theorem \ref{theorem:talagrand} to the
zonotope $K_\mathrm{M}=\sum_{i \in I} \conv \{ \pm M_i \}$ (which lives in a $d^2$-dimensional space), we obtain a zonotope $Z' = \sum_{i \in J} \conv \{ \pm \lambda_i M_i \}$
with $\card J \leq C d^2 \log d /\e^2$ such that $Z' \subset K_{\mathrm{M}} \subset (1+\e) Z'$.
It remains to show that $\mathrm{M}'=(\lambda_i M_i)_{i \in J}$ is a sub-POVM. We know that $h_{Z'} \leq h_{K_\mathrm{M}}$.
Therefore, given a unit vector
$x \in \C^d$, the inequality $h_{Z'}(\Delta) \leq h_{K_\mathrm{M}}(\Delta)$ applied with $\Delta=\ketbra{x}{x}$ shows that
\[  \sum_{i \in J} \lambda_i \left| \langle x | M_i | x \rangle \right| \leq \| \ketbra{x}{x} \|_{\mathrm{M}} \leq \|
\ketbra{x}{x} \|_1 =1, \]
and therefore $\sum_{i \in J} \lambda_iM_i \leq \Id$, as required. Since the inclusions $Z' \subset K_{\mathrm{M}} \subset (1+\e) Z'$ are
equivalent to the inequalities $\|\cdot\|_{\mathrm{M}'} \leq \|\cdot\|_{\mathrm{M}} \leq (1+\e) \|\cdot\|_{\mathrm{M}'}$, Theorem
\ref{theorem:approximation-any} follows.

\section*{Acknowledgements}

We thank Andreas Winter for having first raised the general question of finding POVMs with few outcomes but good discriminating power.
We also thank Marius Junge for suggesting the possible connection between POVMs and zonoids, and for pointing out to us relevant literature.



\end{document}